\documentclass[aps,superscriptaddress,twocolumn,twoside,floatfix,pra,nofootinbib,a4paper]{revtex4}
\usepackage{times}
\usepackage{epsfig}
\usepackage{amsfonts}
\usepackage{amsmath}
\usepackage{amssymb}
\usepackage{amsthm}
\usepackage{color}
\usepackage{multirow}
\usepackage[normalem]{ulem}
\newcommand{\stkout}[1]{\ifmmode\text{\sout{\ensuremath{#1}}}\else\sout{#1}\fi}
\usepackage{latexsym}
\usepackage{mathrsfs}
\usepackage{natbib}
\usepackage{verbatim}
\usepackage[T1]{fontenc}
\usepackage{float}

\usepackage{graphicx}
\usepackage{xcolor}

\usepackage[colorlinks=true,linkcolor=blue,citecolor=magenta,urlcolor=blue]{hyperref}

\DeclareMathOperator{\Tr}{tr}

\newtheorem{theorem}{Theorem}
\newtheorem{corollary}{Corollary}
\newtheorem{conjecture}{Conjecture}

\newcommand{\ket}[1]{|#1\rangle}

\newcommand{\ketbra}[2]{|#1\rangle\langle#2|}

\begin{document}
	
	
\title{Measurement incompatibility and steering are necessary and sufficient for operational contextuality}


\author{Armin Tavakoli}\thanks{A. T. and R. U.  contributed equally to this work.}
\affiliation{D\'epartement de Physique Appliqu\'ee, Universit\'e de Gen\`eve, CH-1211 Gen\`eve, Switzerland}

\author{Roope Uola}\thanks{A. T. and R. U.  contributed equally to this work.}
\affiliation{D\'epartement de Physique Appliqu\'ee, Universit\'e de Gen\`eve, CH-1211 Gen\`eve, Switzerland}

\begin{abstract}
Contextuality is a signature of operational nonclassicality in the outcome statistics of an experiment. This notion of nonclassicality applies to a breadth of physical phenomena. Here, we establish its relation to two fundamental nonclassical entities in quantum theory; measurement incompatibility and steering. We show that each is necessary and sufficient the failure of operational contextuality. We exploit the established connection to contextuality to provide a novel approach to problems in measurement incompatibility and steering. 
\end{abstract}
	
		
\maketitle
	
	
\section{Introduction}
The nonclassical nature of quantum theory has a variety of different manifestations. On the one hand, quantum theory postulates theoretical entities with properties that lack a counterpart in classical physics. On the other hand, the nonclassicality of quantum theory is also present on the observable level, i.e. in the outcome statistics of experiments. Evidently, if an experiment takes the reality of the quantum formalism for granted, every nonclassical entity of quantum theory can be experimentally detected.  However, if the assumption of nature being quantum is dropped, the outcome statistics can frequently  be reproduced with some classical model. Matters become more interesting when the nonclassicality of outcome statistics can be operationally determined in the spirit of device-independence \cite{Focus}. That is, in experiments that demonstrate nonclassicality while making weak assumptions on the underlying physical nature.

The strongest form of operational inference is encountered in tests of Bell inequalities \cite{BellReview}. These experiments statistically analyse the correlations between the outcomes of measurements performed in space-like separated events. If the correlations violate a Bell inequality, it follows that the outcome statistics cannot be explained by \textit{any} classical (local hidden variable) theory. Famously, by sharing entangled states and performing incompatible measurements that together steer the remote partner system, quantum theory can  violate these inequalities and therefore provide an unequivocal demonstration of nonclassicality \cite{Bell, CHSH}. Surprisingly however, not all incompatible measurements, nor all steerable ensembles enable Bell inequality violations \cite{Steering, Tamas, Flavien}. This motivates the question: is nonclassicality at the level of theoretical entities both necessary and sufficient for some form of operational nonclassicality?

We show that two fundamental theoretical  entities in quantum theory always enable a proof of operational nonclassicality. Specifically, we consider the incompatibility of quantum measurements and the ability to steer another system by local measurements and classical communication (a feature of quantum theory originating from Schr\"odinger's remarks \cite{Schrodinger} on the Einstein-Podolsky-Rosen paradox \cite{EPR}). These two features of quantum theory are thoroughly researched, see e.g.~Refs~\cite{Teiko, QtheoryMeas, OpQPhys, QuantumMeasurement} and Refs~\cite{SteeringReview1, SteeringReview2} respectively. We establish their relation to operational nonclassicality as provided by outcome statistics that demonstrate \textit{contextuality}. As originally introduced by Bell, Kochen and Spekker \cite{BellContext, KochenSpekker}, contextuality is a property of projective measurements in quantum theory. However, the concept has seen a generalisation that applies on the level of ontological models, and therefore to general operational theories used to model outcome statistics \cite{Sp05}. Briefly, contextuality is the impossibility of building an ontological model  that ascribes equivalent physical representation to operationally indistinguishable physical procedures. Said impossibility is referred to as preparation (measurement) contextuality when the procedures correspond to ways of preparating (measuring) a system. 
Contextuality (in its broad sense) in quantum theory is closely related to e.g.~advantages in quantum computation \cite{Raussendorf, Delfosse, Magic, Vega}, advantages in particular communication tasks \cite{POM, HT17, Saha, SS18}, quantum zero-error communication \cite{ZeroError} and anomalous weak values \cite{Pusey}.

This work has two main results. Firstly, we show that incompatibility of quantum measurements is necessary and sufficient for enabling a demonstration of nonclassicality in terms of preparation contextual outcome statistics. Secondly, we show that every no-signaling ensemble of states is steerable if and only if it enables a demonstration of nonclassicality in terms of preparation and measurement contextual outcome statistics. These results allow us to view problems in joint measurability and steering through the lens of contextuality. We exploit this to present a family of noncontextuality inequalities and provide numerical evidence that these are necessary and sufficient conditions for certifying the incompatibility of any set of binary qubit observables, and that they also constitute optimal tests of the steerability of a pair of qubits in a singlet state subject to noisy environments. Such applications also make possible experimental tests of measurement incompatibility and steerability  based on  operational contextuality.

\section{Preliminaries}
We begin with a brief summary of contextuality (in the generalised sense of Ref~\cite{Sp05}), measurement incompatibility and quantum steering.

\subsection{Contextuality}
Contextuality in an operational theory  \cite{Sp05} is developed within the framework of ontological models \cite{Harrigan}. An ontological model describes a preparation procedure $\mathbf{P}$ by an ontic state (hidden variable) $\lambda$ with a distribution $p(\lambda|\mathbf{P})$. When a measurement procedure $\mathbf{M}$ is applied, the ontological model determines the probability of an outcome $b$ with some response function $p(b|\mathbf{M},\lambda)$. Therefore, the outcome statistics reads
\begin{equation}\label{prep}
p(b|\mathbf{P},\mathbf{M})=\sum_\lambda p(\lambda|\mathbf{P})p(b|\mathbf{M},\lambda).
\end{equation}
Furthermore, ontological models are linear in the sense that convex combinations of preparation and measurement procedures are represented by convex sums of the relevant ontic state distributions and response functions respectively. See Ref.~\cite{Sp05} for a discussion of this property.

In quantum theory preparations are represented by density matrices $\rho$ and measurements are positive operator-valued measures (POVMs) $M=\{M_b\}$, i.e. $M_b\geq 0$ and $\sum_b M_b=\openone$. Outcome statistics is given by the Born rule  $p(b|\rho,M)=\Tr\left[\rho M_b\right]$. A state (measurement) can be realised in as many ways as it can be decomposed into mixtures of other states (decomposed into element-wise mixtures or coarse-grainings of other measurements). Different ways of preparing the same state (performing the same measurement) are called contexts for $\rho$ ($M$). An ontological model is said to be preparation noncontextual if the ontic state distribution is independent of the context, i.e. if $p(\lambda\lvert \mathbf{P})=p(\lambda\lvert \rho)$.  Similarly, an ontological model is said to be measurement noncontextual if the response functions are context-independent, i.e. if $p(b|\mathbf{M},\lambda)=p(b\lvert M, \lambda)$. These notions embody the idea that if two laboratory procedures are indistinguishable, then they are also indistinguishable on the level of ontic states. We remark that in order to ensure that two procedures truly are indistinguishable, one needs to be able to perform measurements (prepare states) that span the measurement (state) space.  
In contrast, if outcome statistics cannot be reproduced with any preparation (measurement) noncontextual model, it is said to be preparation (measurement) contextual. See Ref~\cite{Sp05} for a detailed discussion of operational contextuality.

\subsection{Measurement incompatibility }
Measurement incompatibility \cite{Teiko, QtheoryMeas, OpQPhys, QuantumMeasurement} is the impossibility of jointly measuring a set of (at least two) POVMs by employing only a single  measurement and classical post-processing of its outcomes. More precisely, let $\{A_{a|x}\}$ be a set of POVMs, with $a$ labelling the outcome and $x$ labelling the measurement. The set is called compatible (jointly measurable) if there exists a POVM $\{G_\lambda\}$ which allows us to recover the set  $\{A_{a\lvert x}\}$ via some post-processing probability distribution $p(a|x,\lambda)$;
\begin{equation}\label{JM}
A_{a|x}=\sum_{\lambda}p(a|x,\lambda)G_\lambda.
\end{equation}
If such a model does not exist, the set $\{A_{a|x}\}$ is called incompatible (not jointly measurable). This extends the textbook concept of commutativity in the sense that mutually commuting POVMs are jointly measurable, but the converse does not hold in general. The converse holds, however, for textbook observables, i.e. projective measurements. It is worth noting that joint measurability can be characterised as the existence of a common Naimark dilation in which the projective measurements commute.

\subsection{Steering}
Steering \cite{Steering} is a qualitative property of some entangled quantum states regarding the set of ensembles that can be remotely prepared with local measurements and classical communication. Specifically, one considers a pair of entangled systems in state $\rho$ and performs a set of measurements $\{A_{a|x}\}$ on the first system. Given the choice of $x$, this renders the second system in the state $\rho_{a|x}=\Tr_\text{A}\left[A_{a|x}\otimes \openone \rho\right]/\Tr\left[A_{a|x}\otimes \openone \rho\right]$ with probability $p(a|x)=\Tr\left[A_{a|x}\otimes \openone \rho\right]$. It is important to underline the fact that classical communication is necessary for the steered party to be able to distinguish between different local states $\rho_{a|x}$. These local states can be effectively described with a the set of unnormalised states (called an assemblage) $\{\sigma_{a|x}\}$ where $\sigma_{a|x}=\Tr_\text{A}\left[A_{a|x}\otimes \openone \rho\right]$. Such assemblages are no-signaling, i.e. $\sum_a\sigma_{a|x}=\sum_a\sigma_{a|x'}$. In this work all assemblages are assumed to be no-signaling. We remark that the Gisin-Hughston-Josza-Wootters theorem \cite{Gisin, HJW} ensures that every assemblage can be prepared by a distant party's local measurements (supported by classical communication) on a properly chosen entangled state.  An assemblage is said to be unsteerable if it admits a so-called local hidden state model. Such models use $(a,x)$ as information towards a post-processing $p(a|x,\lambda)$ of a set of local states $\rho_\lambda$ appearing with probability $p(\lambda)$ to explain the assemblage $\{\sigma_{a|x}\}$. Hence, if the state is unsteerable, it can be written as
\begin{equation}\label{LHS}
\sigma_{a|x}=\sum_{\lambda} p(\lambda)p(a|x,\lambda)\rho_\lambda.
\end{equation}
If no local hidden state model is possible, the assemblage is called steerable.

\section{Main results}
We begin by proving a one-to-one relation between  measurement incompatibility and preparation contextuality. 
\begin{theorem}\label{thm1}
	A set of measurements is compatible if and only if their statistics admit a preparation noncontextual model for all states.
\end{theorem}
\begin{proof}
Assume that the set of POVMs $\{A_{a|x}\}$ when applied to any quantum state $\rho$ returns outcome statistics that is preparation noncontextual. We denote the set of preparation procedures (contexts) in which $\rho$ can be prepared by $\mathcal{P}_\rho$. Then, using the label $x$ to denote the measurement procedure, it holds that  
\begin{equation}\label{step1}
\forall \mathbf{P}\in\mathcal{P}_\rho: p(a\lvert x,\mathbf{P})=\sum_\lambda p(\lambda|\rho)p(a|x,\lambda).
\end{equation}
For each $\lambda$, the object $p(\lambda\lvert \rho)$ is a convexity-preserving map from the space of quantum states to the interval $[0,1]$.  The Riesz representation theorem \cite{QuantumMeasurement, MLQT} asserts that such maps can be written as an inner product $p(\lambda|\rho)=\Tr\left[G_\lambda \rho\right]$ for some unique operator $0\leq G_\lambda\leq \openone$. Moreover, since $\forall \rho: \sum_{\lambda}p(\lambda|\rho)=1$, it follows that $\sum_\lambda G_\lambda=\openone$. Inserting this into Eq.~\eqref{step1}, we have
\begin{equation}\label{jmstat}
\forall \mathbf{P}\in\mathcal{P}_\rho: p(a\lvert x,\mathbf{P})=\sum_\lambda p(a|x,\lambda) \Tr\left[G_\lambda \rho\right].
\end{equation}
We have recovered the outcome statistics obtained from measuring $\rho$ with a compatible set of POVMs. 

Conversely, assume that $\{A_{a|x}\}$ is a compatible set of POVMs. Then, the statistics obtained from measuring any state $\rho$ prepared with a procedure $\mathbf{P}$ is given by Eq.~\eqref{jmstat}. By defining $p(\lambda|\rho)=\Tr\left[G_\lambda \rho\right]$, we recover the definition of outcome statistics being preparation noncontextual.
\end{proof}

It is interesting to note that tests of preparation contextuality  can be formulated as communication tasks between two separated parties, in which the receiver is kept oblivious about parts of the sender's input \cite{POM, Saha, HT17}. Such obliviousness corresponds to different contexts for the states. From theorem~\ref{thm1}, we can therefore infer that:
\begin{corollary}
Every set of incompatible measurements enables a quantum-over-classical advantage in a communication task.
\end{corollary}
We remark that the advantages of all incompatible sets of  measurements have recently been shown in various measurement-device-independent communication tasks \cite{Paul, RoopeN, Carmeli, Leo}.

In a spirit similar to that of theorem~\ref{thm1}, we prove a one-to-one relation between steering and measurement contextuality.
\begin{theorem}\label{thm2}
	An assemblage is unsteerable if and only if its statistics admits a preparation and measurement noncontextual model for all measurements.
\end{theorem}
\begin{proof}
Assume that the assemblage $\{\sigma_{a\lvert x}\}$ when measured with any POVM $M$ returns outcome statistics that is measurement noncontextual. We denote the set of measurement procedures (contexts) in which $M$ can be realised by $\mathcal{M}_M$. Due to assemblages being no-signaling, we have that $p(a,b\lvert x,\mathbf{M})=p(b\lvert a,x,\mathbf{M})p(a\lvert x)$ and that  	
\begin{multline}\label{step2}
\forall \mathbf{M}\in \mathcal{M}_M: p(b\lvert a,x,\mathbf{M})p(a\lvert x)=\\
p(a\lvert x) \sum_{\lambda}p\left(\lambda\lvert a,x\right)p(b\lvert M,\lambda),
\end{multline}
where $(a,x)$ labels the preparation procedure. For every $\lambda$, the object $p(b\lvert M,\lambda)$ is a map from the space of POVMs to the space of probability distributions. Such maps are characterised by  the works of Gleason \cite{Gleason} and Busch \cite{BuschsTheorem}. The Gleason-Busch theorem asserts that $p(b\lvert M,\lambda)=\Tr\left[\rho_\lambda M_b\right]$ for some unique state $\rho_\lambda$. Inserting this into Eq.~\eqref{step2}, we have
\begin{multline}\label{step3}
\forall \mathbf{M}\in \mathcal{M}_M: p(b\lvert a,x,\mathbf{M})p(a\lvert x)\\
=p(a\lvert x)\sum_{\lambda}p\left(\lambda\lvert a,x\right)\Tr\left[\rho_\lambda M_b\right]
\end{multline}
Using   Bayes' rule and the fact that $x$ and $\lambda$ are independent\footnote{The independence of $x$ and $\lambda$ follows from the fact that $\lambda$ cannot carry information about the oblivious variable $x$ (the obliviousness comes from no-signaling), i.e.~the assumption of preparation noncontextuality. See e.g.~Refs.~\cite{POM, HT17} for an elaboration.}, one straightforwardly finds that  $p(a\lvert x)p(\lambda\lvert a,x)=p(\lambda)p(a\lvert x,\lambda)$. Inserting this in \eqref{step3}, we recover the outcome statistics obtained from applying $M$ to an unsteerable assemblage \eqref{LHS}. 

Conversely, if the assemblage has a local hidden state model, then for every POVM the outcome statistics reads
\begin{multline}\label{xxx}
\forall \mathbf{M}\in \mathcal{M}_M: p(b\lvert a,x,\mathbf{M})=\\
\frac{1}{p(a\lvert x)}\sum_{\lambda}p(\lambda)p(a\lvert x,\lambda)\Tr\left[\rho_\lambda M_{b}\right].
\end{multline}
From Bayes rule and the independence of $x$ and $\lambda$, we have that $p(\lambda)p(a\lvert x,\lambda)/p(a\lvert x)=p(\lambda\lvert a,x)$. Note that said independence implies preparation noncontextuality. Inserted into Eq.~\eqref{xxx} we find the outcome statistics obtained in a measurement noncontextual model. 
\end{proof}

We have illustrated the theorems in Figure~\ref{Fig}. Notice that theorem~\ref{thm1} and theorem~\ref{thm2} give a characterisation of the ontic variables using quantum theory. Whereas this characterisation is relevant for noncontextual models covering all states or measurements, it would be interesting to see whether such characterisation exists in the case of fragments of quantum theory, i.e.~for noncontextual models covering subsets of states and measurements. 

\begin{figure}
	\centering
	\includegraphics[width=\columnwidth]{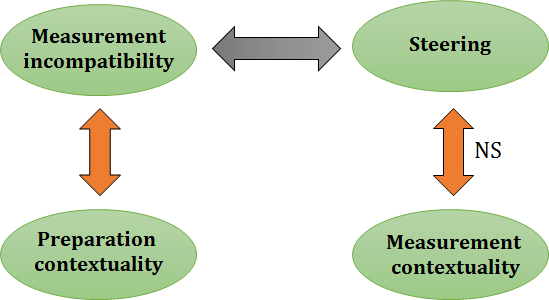}
	\caption{The orange arrows illustrate theorem~\ref{thm1} and theorem~\ref{thm2}. The upper two figures represent theoretical entities which respectively enable the lower two operational features when provided a correct catalyst preparation or measurement procedure. Conversely, the two operational features imply the respective theoretical entities regardless of the underlying (no-signaling, ``NS") procedure. The grey arrow indicates the previously known relation that a set of measurements is incompatible if and only if it enables steering with a proper catalyst state \cite{QV14, Roope2}. It is worth noting that this connection can also be seen as a mapping between the problems in measurement incompatibility and steerability \cite{UB15}.}\label{Fig}

\end{figure}

Also, it is worth noting that a number of works have (in different ways) shown that outcome statistics that violate a Bell inequality is a proof of preparation contextuality \cite{Leifer, HT17, Saha, Sp05}. In Appendix~\ref{LHVtoNC}, we note that this fact follows immediately from ontological models and the no-signaling principle (see also Ref.~\cite{Matt} for similar result)\footnote{This result was shown originally in the unpublished note Ref.~\cite{Barrett}.}.

\section{Noncontextuality inequalities for qubit measurement incompatibility and steering}
We proceed to use the established connection to contextuality to address two relevant problems in measurement incompatibility and steering. I) Can one find a preparation noncontextuality inequality whose violation is both necessary and sufficient for certifying measurement incompatibility for interesting families of measurements? Note that despite theorem~\ref{thm1} this is a nontrivial matter since any \textit{single} test of preparation contextuality is only a sufficient condition for measurement incompatibility (and the full characterisation of all tests of preparation contextuality is a demanding problem). II) Can one find a measurement noncontextuality inequality (in a no-signaling scenario) that for an interesting class of states optimally certifies their steerability? In analogy to the previous, this is nontrivial since any single test of measurement contextuality in a no-signaling scenario is only a sufficient condition for steering.

 To answer these questions, we present a family of correlation inequalities (parameterised by an integer $n\geq 2$) inspired by the works of Refs~\cite{POM, Kerenidis, SpatialSequential}.  Consider a Bell-like (no-signaling) experiment in which two separated observers, Alice and Bob, share parts of a physical system. Alice (Bob) performs  measurements labelled by her (his) uniformly random input $x\in\{0,1\}^{n-1}$ ($y\in\{1,\ldots, n\}$). 
The outcome is denoted by $a\in\{0,1\}$ ($b\in\{0,1\}$).  Alice's measurement procedures are constrained by operational equivalences. That is, her outcome statistics always upholds suitable indistinguishability relations which enable us to consider the statistics of different contexts for her measurements. Specifically, for every bit-string $r\in\{0,1\}^n$, we require that the measurement procedures $\mathbf{M}_{r,0}$ and $\mathbf{M}_{r,1}$ corresponding to a uniform mixing of all $(a,x)$ satisfying $r\cdot \bar{x}=0$ and $r\cdot \bar{x}=1$ respectively (where  $\bar{x}=(a,x+a)$), are  indistinguishable from each other. In quantum theory, this means that 
\begin{equation}\label{op}
\sum_{a,x\lvert r\cdot \bar{x}=0} M_{a\lvert x}=\sum_{a,x\lvert r\cdot \bar{x}=1} M_{a\lvert x}.
\end{equation}
 Note that whenever $r$ has an even number of ones, this condition is always satisfied since $M_{0\lvert x}+M_{1\lvert x}=\openone$. For odd strings $r$, Eq.~\eqref{op} is a nontrivial constraint. Now, let Alice and Bob play a game in which they aim to maximise the probability of finding $a+b=\bar{x}_y\mod{2}$. When Alice is considered the sender of Bob's remotely prepared local states, we can consider the scenario as a test of preparation contextuality. In contrast, when Bob is considered the sender of Alice's remotely prepared local states, we can consider the scenario as a test of preparation and measurement contextuality. In case of either being noncontextual, the average success probability is bounded by 
\begin{equation}\label{NC}
\mathcal{A}_n\equiv \frac{1}{n2^{n-1}}\sum_{x,y} p(a+b=\bar{x}_y\lvert x,y)\leq\frac{n+1}{2n}.
\end{equation}
The proof of this result is a simple modification of the arguments presented in Ref~\cite{POM} and is discussed in Appendix~\ref{AppSpekkens}. A violation of the inequality \eqref{NC} means that Bob's measurements (which are unconstrained) are incompatible (by theorem~\ref{thm1}) and that Alice's local assemblage (prepared by Bob) is steerable (by theorem~\ref{thm2}). We now study the usefulness of the inequality \eqref{NC} for certifying qubit measurement incompatibility and two-qubit steerability. 

For the case of $n=2$ the inequality \eqref{NC} reduces to the Clauser-Horne-Shimony-Holt Bell inequality \cite{CHSH} for which it is known that all pairs of incompatible measurements enable a violation \cite{Wolf}. For $n> 2$ (specifically studying $n=3,\ldots, 7$) we have numerically obtained support (10000 examples for each $n$) for the following conjecture 
\begin{conjecture}\label{conj1}
Every set of $n$ incompatible two-outcome qubit measurements enable a proof of preparation contextuality by a violation of the inequality \eqref{NC}.
\end{conjecture}
In Appendix~\ref{AppNum}, we describe the numerical procedure employed to motivate this conjecture. 

Consider now the case of steering. For simplicity, let Alice and Bob share the noisy singlet state $\rho_v=v\ketbra{\psi^-}{\psi^-}+(1-v)\openone/2$, where $\ket{\psi^-}=(\ket{01}-\ket{10})/\sqrt{2}$ for some visibility $v\in[0,1]$. What is the critical value of $v=v_n$ so that Bob can steer Alice using $n$ projective measurements? Although this question is well-studied (see e.g.~Refs~\cite{Bavaresco, Wiseman, Moroder}) an analytical formula is lacking. However Ref~\cite{Bavaresco} presented nearly matching upper and lower bounds on $v_n$ for $n=2,\ldots, 13$ and $n=2,\ldots, 5$ respectively. Using our inequality \eqref{NC}, we have numerically implemented alternating convex searches to find an upper bound on the critical $v_n$ (below which we can no longer find a quantum violation). This returns
\begin{align}\nonumber
& v_2=0.7071 & v_3=0.5774 && v_4=0.5547 \\
& v_5=0.5422 & v_6=0.5270 && v_7=0.5234.
\end{align} 
Interestingly, these numbers coincide precisely with those presented in Ref~\cite{Bavaresco} (up to the number of decimals presented in Ref~\cite{Bavaresco}). This motivates the conjecture
\begin{conjecture}\label{conj2}
	The inequality \eqref{NC} is a tight steering inequality for the noisy singlet state under $n$ projective measurements.
\end{conjecture}

Finding a conclusive proof of conjectures \ref{conj1} and \ref{conj2} would be interesting.  We remark that although the above considerations are straightforwardly analysed with a computer, the criterion \eqref{NC} can be treated in a  fully analytical manner.

\section{Discussion}
We have shown that every set of incompatible measurements and every steerable assemblage can be operationally certified as nonclassical in a test of operational contextuality, and that the latter also imply the formers. A direct consequence is that problems of joint measurability and steering can be viewed through the lens of contextuality, as we examplified through our conjectures. In this sense, our results bridge the two research directions of quantum measurements and quantum steering with the line of research focused on quantum contextuality. 

Moreover, since tests of operational contextuality only rely on weak characterisation of the experimental devices \cite{Mazurek}, our results can also be considered as semi device-independent certificates of measurement incomaptibility and steering. Naturally, fully device-independent certificates are found by violating a Bell inequality. However, in addition to such tests being experimentally demanding, it is importantly also the case that not all incompatible measurements nor all steerable ensembles violate any Bell inequality \cite{Steering, Tamas, Flavien}. This makes tests of operational contextuality relevant for practical considerations when no fully device-independent certificate is either possible or known. 

\textit{Acknowledgements.---} We thank Alastair Abbott, Costantino Budroni, Matthew Pusey, Tom Bullock, Chau Nguyen, Marco Quintino and Leonardo Guerini for comments. This work was supported by the Swiss National Science Foundation (Starting grant DIAQ, NCCR-QSIT) and by the Finnish Cultural Foundation.

\appendix

\section{Bell nonlocality implies preparation contextuality}\label{LHVtoNC}
We give a simple argument for every probability distribution that violates a Bell inequality also being a proof of preparation contextuality (see also Ref.~\cite{Matt}). We show this immediately from ontological models supplemented with the no-signaling principle encountered in Bell inequality tests.

To see this, we write a general ontological model for a Bell experiment as
\begin{equation}
p(a,b|x,y)=\sum_{\lambda} p(a|x,y)p(\lambda|a,x)p(b|y,\lambda).
\end{equation}
If we also impose no-signaling, then Alice's local marginals are independent of Bob's input. Therefore,
\begin{equation}\label{ontic}
p(a,b|x,y)=\sum_{\lambda}p(a|x)p(\lambda|a,x)p(b|y,\lambda).
\end{equation}
Bayes' rule together with the independence of $x$ and $\lambda$ give that $p(a|x)p(\lambda|a,x)=p(\lambda)p(a|x,\lambda)$. Inserting this into Eq.~\eqref{ontic}, we obtain 
\begin{equation}
p(a,b|x,y)=\sum_{\lambda}p(\lambda)p(a|x,\lambda)p(b|y,\lambda).
\end{equation}
This is a local hidden variable model, i.e.~the notion of classicality in Bell inequality tests. The assumption of preparation noncontextuality is enforced due to the assignment of the same ontic-state distribution for the preparation procedures corresponding to the remotely prepared state on Bob's side when averaged over Alice's outcomes, i.e.~the principle of no-signaling. Therefore, whenever $p(a,b|x,y)$ has no local hidden variable model, it also has no preparation noncontextual model. 

\section{Noncontextual bound in inequality \eqref{NC}}\label{AppSpekkens}
In the main text, we considered a scenario in which separated parties Alice and Bob share a state and perform local measurements with binary outcomes $a,b\in\{0,1\}$. Alice's measurement settings are labelled by a bit-string $x\in\{0,1\}^{n-1}$ and Bob's measurement settings are labelled by $y\in\{1,\ldots,n\}$. Alice and Bob aim to satisfy the relation $a+ b=\bar{x}_y\mod{2}$ where $\bar{x}=(a,a+ x)$ is an $n$-bit string. The notation $\bar{x}_y$ labels the $y$'th bit in the string $\bar{x}$. Their average success probability is
\begin{equation}\label{NC2}
\mathcal{A}_n\equiv \frac{1}{n2^{n-1}}\sum_{x,y} p(a+ b=\bar{x}_y\lvert x,y).
\end{equation}
Alice and Bob are restricted by two constraints. Firstly, they obey the no-signaling principle. This means that the preparations of Alice on Bob's side (denoted $\mathbf{P}_{a,x}$), effectively  achieved  by a local measurement on her system, realise the same preparation in different contexts. That is, the following operational equivalences hold; $\sum_a \mathbf{P}_{a,x}\sim \sum_a \mathbf{P}_{a,x'}$. The analogous holds in the other direction, i.e.~by the preparations of Bob on Alice's side achieved by him locally measuring his system. Secondly, Alice's measurements are required to uphold certain operational equivalences. In quantum theory, these are written 
\begin{equation}\label{op2}
\sum_{a,x\lvert r\cdot \bar{x}=0} M_{a\lvert x}=\sum_{a,x\lvert r\cdot \bar{x}=1} M_{a\lvert x},
\end{equation}
for every $n$-bit string $r\in\{0,1\}^n$ with at least two instances of $'1'$. For clarity, we examplify the case of $n=3$. There exists eight three-bit strings of which four have at least two instances of $'1'$.  Those are $r=011$, $r=101$, $r=110$ and $r=111$. For each $r$ we have the relation in Eq.~\eqref{op2}. In case of, for example, $r=011$ we find
\begin{equation}
M_{0|00}+M_{1|00}+M_{0|11}+M_{1|11}=M_{0|01}+M_{1|01}+M_{0|10}+M_{1|10}.
\end{equation} 
However, this is trivially satisfied since $\forall x: M_{0|x}+M_{1|x}=\openone$. Similarly, one finds that the constraint \eqref{op2} is trivial also for $r=101$ and $r=110$. However, for $r=111$ we obtain 
\begin{equation}
M_{0|00}+M_{0|11}+M_{1|01}+M_{1|10}=M_{0|01}+M_{0|10}+M_{1|00}+M_{1|11},
\end{equation}
which is a nontrivial constraint.

Imagine now that instead of performing local measurements on a shared state, Alice directly prepares the would-have-been post-measurement states of Bob's system (labelled by the pair $(a,x)$) and sends them to Bob, who then measures the system and records $b\in\{0,1\}$. This represents a prepare-and-measure scenario in which Alice has $2^n$ inputs $(a,x)$ with some prior distribution $p(a,x)=p(a|x)/2^{n-1}$. 
Alice's preparations are required to satisfy the operational equivalence which in quantum theory reads
\begin{equation}\label{rep}
\forall r: \quad \sum_{a,x\lvert r\cdot \bar{x}=0} p(a|x)\rho_{a,x}=\sum_{a,x\lvert r\cdot \bar{x}=1} p(a|x)\rho_{a,x}.
\end{equation}
Notice first that in the original scenario, every assemblage prepared by Alice on Bob's side can also be directly sent in this prepare-and-measure model; simply define $\rho_{a,x}=\Tr_\text{A}\left[M_{a|x}\otimes \openone \rho\right]/\Tr\left[M_{a|x}\otimes \openone \rho\right]$, and the prior distribution as $p(a|x)=\Tr\left[M_{a|x}\otimes \openone \rho\right]$. Conversely, every ensemble that Alice can communicate to Bob in the prepare-and-measure scenario can also be realised in the original scenario via local measurements on an entangled state and classical communication. This follows from the Gisin-Hughston-Josza-Wootters theorem \cite{Gisin, HJW} and the fact that Eq.~\eqref{rep} enforces a no-signaling-like preparation ensemble.

In Ref.~\cite{POM} it was shown that when $p(a|x)=1/2$, the considered prepare-and-measure scenario serves as the following test of preparation contextuality; the inequality
\begin{equation}
\frac{1}{n2^{n-1}}\sum_{a,x,y}p(a|x)p(b=(a,x)_y|a,x,y)\leq \frac{n+1}{2n}
\end{equation}
holds for every preparation noncontextual model. Moreover, it is a trivial modification of the arguemtns of Ref.~\cite{POM} to show that the same bound holds regardless of the prior distribution $p(a|x)$. Therefore, due to the  between the prepare-and-measure scenario and the original scenario, it also holds that
\begin{equation}
\mathcal{A}_{n}\leq \frac{n+1}{2n}
\end{equation}
in a preparation noncontextual model, in which we view Alice as effectively preparing the local states of Bob. 

Moreover, in the original scenario, we can equally well consider Bob as the effective sender of Alice's local states. If we impose measurement noncontextuality, the response function of Alice takes no regard to the different contexts of her measurement (related to $r$). Note that preparation noncontextuality is still present due to Alice and Bob being no-signlaing.

\section{Numerical evidence in support of conjecture \ref{conj1}}\label{AppNum}

The numerical evidence behind the conjecture \eqref{conj1} was obtained as follows. We used the prepare-and-measure variant (discussed in the previous Appendix, based on Ref~\cite{POM}) for the numerics. We sample a set of $n$ random two-outcome qubit POVMs $\mathcal{M}=\{B_{b|y}\}_{y=1}^n$. The sampling is done by using the Bloch sphere parameterisation of the most general two-outcome qubit measurement, i.e.
\begin{align}
\forall y: \hspace{5mm} & B_{0| y}=\frac{\alpha_y\openone+\eta_y \vec{n}_y\cdot \vec{\sigma}}{2}\\
 & B_{1| y}=\frac{(2-\alpha_y)\openone-\eta_y \vec{n}_y\cdot \vec{\sigma}}{2}
\end{align}
for some random unit vectors $\vec{n}_y$, some random numbers $\eta_y\in[0,1]$ and some random numbers $\eta_y\leq \alpha_y\leq 2-\eta_y$. 

For the sampled $\mathcal{M}$, we evaluate the largest possible value of the witness $\mathcal{A}_n$ via a semidefinite program optimising over the state ensemble of Alice. This returns the maximal value of $\mathcal{A}_n(\mathcal{M})$ attainable with $\mathcal{M}$. We denote the optimal ensemble returned by the semidefinite program by $\mathcal{P}$. Provided that  $\mathcal{A}_n(\mathcal{M})$ violates the inequality Eq.~\eqref{NC} (in its prepare-and-measure form), we construct new measurements $B'_{b|y}=vB_{b|y}+(1-v)\openone/2$ where $v\in[0,1]$. We write $\mathcal{M}'=\{B'_{b|y}\}$. For the states $\mathcal{P}$ we have that 
\begin{equation}
\mathcal{A}_n(\mathcal{M}',\mathcal{P})=v\mathcal{A}_n(\mathcal{M})+(1-v)\mathcal{A}_n\left(\{\openone/2\},\mathcal{P}\right).
\end{equation}
We choose the value of $v$ for which $\mathcal{A}_n(\mathcal{M}',\mathcal{P})$ saturates the noncontextual bound in Eq.~\eqref{NC}, i.e.
\begin{equation}
v=\frac{\mathcal{C}_n-\mathcal{A}_n\left(\{\openone/2\},\mathcal{P}\right)}{\mathcal{A}_n(\mathcal{M})-\mathcal{A}_n\left(\{\openone/2\},\mathcal{P}\right)},
\end{equation}
where $\mathcal{C}_n=(n+1)/(2n)$ is the noncontextual bound. Then, via a semidefinite program, we check whether  $\mathcal{M}'$ is jointly measurable. Evidently, any perturbation of $v$ to the positive renders $\mathcal{M}'$ incompatible since it implies a violation of the preparation noncontextuality inequality. We have repeated the procedure 10000 times (post-selected on the cases in which $\mathcal{A}_n(\mathcal{M})$ constitutes a proof of preparation contextuality) for $n=3,4,5,6, 7$ respectively. Without exception, we have found that $\mathcal{M}'$ is jointly measurable.

\end{document}